\documentclass[letterpaper, 10 pt,conference]{ieeeconf}  

\IEEEoverridecommandlockouts                              
\usepackage{booktabs}
\usepackage{booktabs}
\usepackage{tikz}
\usetikzlibrary{matrix}
\usetikzlibrary{arrows,automata}
\usepackage{amsmath,amssymb,url}
\usepackage{algorithm,algpseudocode}
\usepackage{graphics}
\usepackage{float}
\usepackage[T1]{fontenc}
\usepackage[utf8]{inputenc}
\usepackage{subcaption}
\usepackage{algpseudocode}
\usepackage{tabularx}
\usepackage{booktabs}
\usepackage{enumerate}
\usepackage{verbatim}
\usepackage{pgf}
\newtheorem{theorem}{Theorem}
\newtheorem{lemma}{Lemma}
\newtheorem{proposition}{Proposition}

\newtheorem{definition}{Definition}
\newtheorem{corollary}{Corollary}

\newcommand{\be}{\begin{equation}}
\newcommand{\ee}{\end{equation}}

\newcommand{\argmin}{\mathop{\rm argmin}}

\usepackage{relsize}
\usepackage{exscale}

\usepackage{setspace}
\usepackage{color}
\definecolor{light-gray}{gray}{0.95}

\long\def\old#1{}

\def \cut{\text{c}}
\def \homega{\hat{\omega}}

\title{\LARGE \bf
A lower bound on the performance of dynamic curing policies for epidemics on graphs}

\author{Kimon Drakopoulos \and Asuman Ozdaglar \and John N. Tsitsiklis
\thanks{{The} authors are with the Laboratory of Information and Decision Systems, Massachusetts Institute of Technology, 
Cambridge, MA 02139. Emails: \{kimondr,asuman,jnt\}@mit.edu}
\thanks{{Research partially supported by the ARO MURI W911NF-12-1-0509 and NSF grant CMMI-1234062. }}
}

\begin{document}

\maketitle

\begin{abstract}
 We consider an SIS-type epidemic process that evolves on a known graph. We assume that a fixed curing budget  can be allocated at each instant to the nodes of the graph, towards the objective of minimizing the expected extinction time of the epidemic. We provide a lower bound on the optimal expected extinction time  as a function of the available budget, the epidemic parameters, {the maximum degree,  and} the CutWidth of the graph. For graphs with large CutWidth (close to the largest possible), and under a budget which is sublinear in the number of nodes, our lower bound scales exponentially with the size of the graph. 
\end{abstract}

\section{Introduction}
We consider an SIS{-}type epidemic model with a common infection rate $\beta$ and an endogenous, node specific curing rate $\rho_i(t)$. 
{A network planner has a total curing budget, which is allocated to different   nodes at each point in time according to
a dynamic policy based on the history of the epidemic and  the network structure.} 

In earlier work \cite{DOTtnse14,DOTcdc14} we propose a policy which achieves rapid containment for any set of initially infected nodes under some conditions on the total curing budget. More specifically,  we prove that if the CutWidth of the underlying graph is much smaller than the available curing budget, then rapid containment of any epidemic is achieved  by the proposed policy.

In this paper we study the case where the CutWidth of the graph is larger than the available budget. Specifically, we focus on a class of graphs for which the CutWidth of the graph is  close to the largest possible value and obtain a lower bound on the optimal expected extinction time. When the curing budget scales sublinearly in the number of nodes, our bound implies that under any dynamic curing policy, the expected time to extinction scales exponentially in the {number of nodes}, in the worst case {where}   all nodes are initially  infected.  {This complements our results in \cite{DOT14}, where we provide (weaker) lower bounds for the case where the} CutWidth is larger than the available budget, linear in the {number of nodes,} but not necessarily close to the largest possible value.

A similar {model, but in which the curing rate allocation is {static}  (open-loop)}
 has been studied in \cite{Cohen,gourdin, chung,preciado}, 
and the proposed methods {were} either heuristic or based on  mean-field approximations of the evolution process. {Closer} to our work, {the authors of  \cite{Borgs10} let the curing rates be proportional} to the degree of each node --- {but again} independent of the current state of the network, which means that curing resources may be wasted on healthy nodes. On a graph with bounded degree, the policy in \cite{Borgs10} achieves sublinear  time to extinction, {but requires a} curing budget that is proportional to the number of nodes. Moreover, when the underlying graph is an expander and the curing budget is sublinear, {it is established in \cite{Borgs10}}  that the optimal expected extinction time scales exponentially in the size of the graph. We extend this result by obtaining an exponential lower bound for all graphs whose CutWidth is close to the largest possible. {This class of graphs contains expander graphs but is substantially larger. } Furthermore, our result also applies to dynamic policies.

The rest of the paper is organized as follows. In Section \ref{model} we present the details of our model. In Section \ref{sec:combinatorics} we study relevant graph-theoretic properties and present {some key combinatorial results that are used in our analysis.}  In Section \ref{sec:exponential} we present the lower bound on the optimal expected extinction time.

\section{The Model} \label{model}
We consider a {network}, represented by a {connected} {undirected} graph  $G=(V,\mathcal{E})$, where $V$ denotes the set of nodes and  $\mathcal{E}$ denotes the set of edges. {We use $n$ to denote the number of nodes.} Two nodes $u,v \in V$ are {\it neighbors} if $(u,v) \in \mathcal{E}$. We denote by $\Delta$ the maximum {of the node degrees.} To exclude 
{trivial cases, and without loss of generality, we assume throughout that $G$ is connected and, in particular,
$\Delta>0$.}

We assume that {the} nodes in a set $I_0{\subseteq V}$ are initially infected and {that} the infection spreads according to a controlled contact process where the rate at which infected {nodes} get cured is  determined by a network controller. Specifically, each {node} can be in one of two states: {\it infected}  or {\it healthy}. The controlled contact process --- {also {known} as the { controlled} SIS epidemic {model}} --- on $G$ is a 
{right-continuous}, continuous-time Markov process 
{$\{I_t\}_{t\geq 0}$ on the state space $\{0,1\}^V$, where $I_t$ stands for the set of infected nodes at time~$t$.} {We refer to $I_t$ as the {\it infection process}.}

State transitions at each node occur independently according to the following dynamics.
\begin{enumerate}[a)] \item {The process 
is initialized at the given initial state $I_0$.}
\item {If a node $v$ is healthy, i.e., if $v\notin I_t$, the transition rate associated with a change of the state of that node to being infected is  equal to an infection rate $\beta$ times the number of infected neighbors of $v$, that is,
$$\beta \cdot  \big|\{(u,v)\in \mathcal{E}: u\in I_t\}\big|,$$
where we use $|\cdot|$ to denote the cardinality of a set.}
{By rescaling time, we can and will assume throughout the paper that $\beta=1$.} 

\item {If a node $v$ is infected, i.e., if $v\in I_t$, the transition rate associated with a change of the state of that node to being healthy is equal to a curing} rate  
 $\rho_v(t)$ that is determined by {the} network controller, as a function of  the current and past states of the {process.} We are assuming here that the network controller has access to {the entire} past evolution of the process.
\end{enumerate}

We assume a {\it budget constraint} 
{of the form}    
\begin{equation}
\sum_{v \in V} \rho_v(t)\leq r,  \label{eq:budgetconstr}
\end{equation}
{at} each time instant $t$, reflecting the fact that curing is costly. 
 A {\it curing policy}  is a  mapping {which at any time $t$ maps the past history of the process to a curing vector $\rho(t)=\{\rho_v(t)\}_{v \in V}$} that satisfies  (\ref{eq:budgetconstr}). The continuous time nature of the problem allows us to restrict out attention to policies that focus on one node at each time, without loss of generality.

We define the {\it  time to extinction} as the {time until the process reaches the absorbing state where all nodes are healthy:} 
\begin{equation}\label{eq:extinctiontime}
\tau = {\min}\{t\geq0 : {I_t=\emptyset}\}. 
\end{equation}
{The {\it expected  time to extinction} (the expected value of $\tau$)  is the performance measure that we will be focusing on.}

\subsection{{A perspective on our main result}}
If the graph $G$ is complete,  all policies that always allocate the entire curing budget {to infected nodes} are essentially equivalent, {in the sense 
that the dynamics of $|I_t|$, the number of infected nodes, are identical under all such policies.}  Furthermore,  {$I_t$} evolves as a birth-death Markov chain which is easy to analyze, and it is not hard to show that the expected time to extinction increases exponentially with $n$. On the other hand, for more general graphs with large CutWidth but bounded degree, an analysis using a one-dimensional birth-death chain or a simple Lyapunov function does not seem possible.

{A related, and conceptually simple, way of deriving lower bounds for more general graphs is to try to show} that the process must make consistent progress through configurations {(subsets $I$ of $V$)} where the total curing rate is significantly lower than the {total} infection rate. Such progress {must then be} a low-probability event, {implying} an exponential lower bound on the time to extinction. Unfortunately, {it is not clear whether} this line of argument, based only on the instantaneous infection rates (namely, the ``cuts'' that are encountered --- see Section \ref{sec:combinatorics} for precise definitions) {is possible for general graphs. Indeed, the lower bounds in \cite{DOT14}, for more general regimes, rely on a much more sophisticated argument.}

{The main technical contribution in this paper is to show that in the regime examined (large CutWidth), the above outlined simple approach to deriving lower bounds is successful. Based on some nontrivial combinatorial properties of the CutWidth, and the related concept of the \emph{resilience} of a set of nodes, we will show that there is a sizeable part of the configuration space   in which $|I_t|$ has a strong upward drift. }

\section{Graph theoretic preliminaries}\label{sec:combinatorics}

In this section we consider a deterministic version of the problem {in which we start with all nodes infected, and then cure them one at a time, deterministically, in a way that minimizes }
the maximum cut encountered in the process. This problem has been studied in various forms in the literature \cite{LaPaugh93,Pars78} and its optimal value is {called}  the \emph{CutWidth} of a graph. 
{Our analysis will be based on a related quantity, the \emph{resilience} of a set of nodes, defined in a similar fashion as the CutWidth, except that we start with a  given arbitrary set of nodes, and we are also allowed to} 
deterministically infect nodes in the process. We study {the} properties of the resilience of a {set} of nodes and relate it to the size of the subset (Section \ref{sec:resil}). {Then, in Section \ref{s:properties}, we relate the resilience and the cut associated to a set of nodes,  which is the key  combinatorial result behind our analysis.}

\subsection{CutWidth and resilience}

We first introduce {some terminology: ``\emph{bags}'' and ``\emph{crusades}.''}

\begin{definition} A \textbf{bag}  is a subset $A\subseteq V$ of the set of nodes $V$. We denote by $|A|$ the number of nodes {in} $A$. \end{definition}

{We} introduce two common {operations} on a bag $A$ and we write 
\[
A+v=A\cup\{v\}, \qquad\text{ for any  } v \notin  A,
\]
and 
\[
A-v=A \setminus \{v\}, \qquad\text{ for any } v \in A.
\] 
We also use the notations
$
A \setminus B= \{v \in A: v \notin B\}
$
{$A  {\blacktriangle} B=(A\setminus B) \cup (B\setminus A)$.}

We next define the concept of a crusade. A crusade from $A$ to $B$ is a sequence of bags that starts from $A$ and ends at $B$ with the restriction that at every step of this sequence arbitrarily many nodes may be added to the previous bag, but at most one can be removed. 

\begin{definition}
For any two bags $A$ and $B$, an  \textbf{$(A-B)$-crusade} $\omega$  is a sequence $(\omega_0,\omega_1,\ldots,\omega_{k})$ of bags, of length $|\omega|=k+1$, with the following properties:
\begin{enumerate}[(i)]
\item $\omega_0 = A$,
\item $\omega_{{k}}=B$, and
\item $|\omega_i \setminus \omega_{i+1} |\leq 1,$ for  $i \in \{0, \ldots, k-1\}.$
\end{enumerate}
\end{definition}
Property (iii) states that at every step of a crusade arbitrarily many nodes can be added to the current bag but \emph{at most one} node may be removed from the bag.
Note that the definition of a crusade allows \emph{non-monotone} moves, since a bag at any step can be a subset, a superset or  not comparable to the preceding bag. 
We denote by $\Omega(A-B)$ the set of all $(A-B)$-crusades. 

We also consider a special case of crusades, the \textbf{monotone  crusades} for which only removal of nodes is allowed at each step. Specifically,  for any two bags $A$ and $B$, $A,B\subseteq V$,  an \textbf{($A \downarrow B $)-monotone crusade} $\omega$   is an ($A-B$)-crusade $(\omega_0,\omega_1,\ldots,\omega_{k})$ with the additional property: $\omega_i \supseteq \omega_{i+1}$, for $i \in \{0, \ldots, k-1\}$.
We denote by $\Omega(A \downarrow B)$ the set of all $(A \downarrow B)$-crusades. 

The number of edges connecting {a} bag {$A$} with its complement is called the cut of the bag. Its importance {comes from the fact that it is equal to the total rate at which new infections occur, when the set of currently infected nodes is~$A$.}
\begin{definition} {For any bag $A$, its \textbf{cut}, $\cut(A)$, is defined as the cardinality of the set of edges
$$\big\{(u,v): u\in A,\ v\in A^c\big\}.$$}
\end{definition}

\noindent 
{In Proposition \ref{prop:cutproperties} below, we record, without proof, four  elementary properties of cuts.}

\begin{proposition}
\label{prop:cutproperties}
For any two bags $A$ and $B$, we have 
\begin{itemize}
\item[(i)]$\cut(A\cup B)\leq \cut(A)+\cut(B)\leq \cut(A)+\Delta\cdot |B|.$
\item [(ii)] If $A\subseteq B$, then  for   any $v \in A $, 
\[
\cut(A-v)  - \cut(A) \leq \cut(B-v) - \cut(B).
\]
\item[(iii)]	 $\cut(A)\leq \min\{|A|\cdot \Delta, (n-|A|)\cdot \Delta\}$
\item[(iv)] {$|\cut(A) -\cut(B) | \leq \Delta |A  {\blacktriangle} B|.$ }
\end{itemize}
\end{proposition}
Note that Proposition \ref{prop:cutproperties}(ii) {states the well-known}  submodularity property of the function $\cut(\cdot)$.

We define the width of a crusade $\omega$ as the maximum cut encountered during the crusade. Intuitively, this is the largest infection rate to be encountered if the nodes were to be cured deterministically according to the sequence prescribed by the crusade (assuming no {new} infections {happen}  in between)  .

\begin{definition}
The \textbf{width} $z(\omega)$ of  an ($A-B$)-crusade $\omega=(\omega_0, \ldots, \omega_{{k}})$ is defined by
\[
z(\omega)=\max_{ {i\in\{1,\ldots,k\}}}\{\cut(\omega_i)\}.
\]
\end{definition}
Note that in the definition above, the maximization starts {after} the first step of the crusade, i.e., we exclude the first bag $\omega_0$ from the maximization. 

We now formally define the CutWidth of a graph $G$ as the minimum over all monotone crusades from $V$ to the empty set  of the corresponding crusade width. 
\begin{definition} For any given graph $G$, its \textbf{CutWidth} $W$ is given by 
\begin{equation}
W=\min_{\omega \in \Omega(V \downarrow \emptyset)} z(\omega). \label{eq:cwdef}
\end{equation}
\end{definition}
 Intuitively, this metric indicates the  maximum cut that is encountered after the first step during {an ``optimal''} monotone crusade {that}  clears the graph. 
 
{The largest possible value of a cut, for graphs with maximum degree $\Delta$,  is $n \Delta /2$, and therefore the CutWidth is also upper bounded by $n \Delta /2$.}
For notational convenience, {we define} 
\begin{equation}
 E=\frac{2}{\Delta}\left( \frac{(n+2) \Delta}{2}-W  \right),
 \label{eq:e}
\end{equation}
{and observe that $E\geq {2}$. 
 {Note that ``small'' values of $E$ indicate that the CutWidth is not too far from the largest possible value, $n\Delta/2$.}
 In Section \ref{sec:combinatorics} we relate {$E$ to cuts} and show that when $E$ is small, then bags with {large} resilience, {as defined below,} also have a {large} cut. 

\begin{definition}  The \textbf{resilience}, {$\gamma(A)$, of a bag $A$ is defined by}    
\begin{equation}
\gamma(A)= \min_{\omega \in \Omega(A-\emptyset)} z(\omega).\label{eq:resiliencedef}
\end{equation}
{We} denote by $\Omega^A \subseteq \Omega(A-\emptyset)$ the set of  crusades $\omega$ that {attain} the minimum, i.e., 
\[
z(\omega)=\gamma(A).
\]
Crusades in $\Omega^A$ are referred to as $A$-optimal, {or simply as optimal, when the set $A$ is clear from the context.}
\end{definition}

Intuitively, $\gamma(A)$ indicates the  maximum cut that is encountered, after the first step, during a crusade {that}  clears a bag $A$. 
Note that in contrast  to the definition of the CutWidth of a graph,  the minimization, in the definition of  resilience, is over  all crusades, not just monotone crusades.

It can be seen that the resilience  of a bag $A$ satisfies the Bellman equation: 
\begin{equation} \label{eq:bellman}
\gamma(A) = \min_{|A\setminus B|\leq 1} \left\{ \max\{\cut(B), \gamma(B)\}\right\}.
\end{equation}

\subsection{Properties of the resilience}\label{sec:resil}
This section explores {some} properties of the resilience.  {We first prove (Lemma \ref{lem:infectiondelta}(i))} that if $A$ and $B$ are two {bags with}  $A \subseteq B$, then $\gamma(A) \leq \gamma(B)$. Intuitively, one can construct a crusade from $A$ to $\emptyset$ as follows:  The crusade starts from  $A$, then continues to the first bag encountered by a $B$-optimal crusade $\omega^B$,  and then follows $\omega^B$. The constructed crusade and $\omega^B$ are the same except for the initial bag of the crusade. By the definition of the resilience, the initial bag does not affect the maximization  and thus the {width of the} new crusade  is equal to $\gamma(B)$.  

We also prove {(Lemma \ref{lem:infectiondelta} {(ii)})} that if two bags differ by only one node $v$, then the corresponding resiliences are at most $\Delta$ apart. Intuitively, one can attach node $v$  to the optimal crusade {for} the {smaller} of the two bags and achieve a maximum cut which is at most $\Delta$ different from the original. 

\begin{lemma}\label{lem:infectiondelta}
{Let $A$ and $B$ be two bags.}
\begin{enumerate}[(i)]
\item \textbf{[Monotonicity]} {If} $A \subseteq B$, then  $\gamma(A) \leq \gamma(B)$.
\item \textbf{[{Smoothness}]} {If} $B=A+v$, then $\gamma(B) \leq \gamma(A)+\Delta$.
\end{enumerate}
\end{lemma}

\begin{proof}

\begin{enumerate}[(i)]
\item  Let $\omega^B=(\omega^B_0, \ldots, \omega^B_k)\in \Omega^B$. Consider the sequence $\hat{\omega}=(\homega_0, \ldots, \homega_k)$ of bags  for  which $\hat{\omega}_0 = A$, and $\hat{\omega}_i =\omega^B_i$, for $i = 1, \ldots, k$.  We claim that $\hat{\omega}$ is a crusade $\hat{\omega}\in \Omega(A - \emptyset)$. Indeed,
 \begin{enumerate}
 \item $\hat{\omega}_0=A$;
 \item $\hat{\omega}_k=\omega^B_k=\emptyset$
 \item $|\hat{\omega}_0\setminus \hat{\omega}_1|=|A\setminus\hat{\omega}_1|\leq|B \setminus{\omega^B_1}|=|{\omega^B_0} \setminus{\omega^B_1}|\leq 1$, where the first inequality follows from  $A \subseteq B$ and $\hat{\omega}_1= {\omega^B_1}$. Moreover, for  $i = 1, \ldots, k$, we have  $|\hat{\omega}_i
 \setminus\hat{\omega}_{i+1}|=|{\omega^B_i}
  \setminus 
 {\omega^B_{i+1}}|\leq 1$.
\end{enumerate}
Clearly,
\[
z(\hat{\omega})=\max_{i \in \{1, \ldots, k\}}    \{\cut(\hat{\omega}_i )\} =\max_{i \in \{1, \ldots, k\}}\{\cut(\omega^B_i )\} = \gamma(B).
\]
\noindent Concluding, by {the} definition {of $\gamma(A)$,}
\[
\gamma(A)= \min_{\omega \in \Omega(A- \emptyset)} z(\omega) \leq z(\homega)= \gamma(B).
\]

\item    Let $\omega^A=(\omega^A_0, \ldots, \omega^A_k)\in \Omega^A$.  Consider the sequence $\hat{\omega} = (\homega_0, \ldots, \homega_{k+1})$ of bags  {with}  $\hat{\omega}_i ={\omega^A_i} \cup \{v\}$, for  $i= 0, \ldots, {k }$, and $\hat{\omega}_{k+1}=\emptyset$. Clearly, $\hat{\omega}$ is a crusade $\hat{\omega} \in \Omega(B - \emptyset)$ for which 
 \[
\gamma(B) \leq z(\homega) \leq \max_{i \in \{1, \ldots, {k }\}}    \{\cut(\omega^A_i )\} + \Delta= \gamma(A) + \Delta,
 \] 
since the addition of one node can  change the cut at each stage of the crusade by at most $\Delta$ (Proposition \ref{prop:cutproperties}(i)).
   \end{enumerate}
\end{proof}

So far the notion of the resilience of  a bag has not been related to the infection process under consideration. Recall that for the infection process defined in Section \ref{model}, the number of infected nodes tends to increase at a rate equal to the   cut of  $I_t$. Therefore, in order to study the evolution of the set of infected nodes $I_t$ through  $\gamma(I_t)$, we relate $\gamma(I_t)$ to $\cut(I_t)$. To this end, we  define a special class of bags, called \emph{improvement} bags. 

\begin{definition}
Let $\mathcal{C}=\{A\subseteq V : \text{ there exists } v \in A \text{ for which } \gamma(A-v)<\gamma(A)\}$. Any bag {in} $\mathcal{C}$ {is called} an {\textbf{improvement}} bag.
\end{definition}

Improvement bags have the important property that their cut can be approximately lower bounded by their resilience. Consequently, whenever an improvement bag with high resilience is encountered, the infection rate is  also high. This observation will play a central role in subsequent sections.


 \begin{lemma} \label{lem:cutwhendrops}
For any  improvement bag $A \in \mathcal{C}$, 
\[
\cut(A) \geq \gamma(A)  -\Delta.
\]
 \end{lemma} 
\begin{proof}
Let $v$ be such that $\gamma(A-v)<\gamma(A)$ and let $B=A-v$.
{Since $|A\setminus B|=1,$  Eq.~\eqref{eq:bellman} yields}
\begin{equation}
\gamma(A) \leq \max \{c(B), \gamma(B)\}. \label{eq:drop}
\end{equation}

{Having assumed that} $\gamma(B)<\gamma(A)$,  Eq.~(\ref{eq:drop}) implies that $\gamma(B)<c(B)$. Hence, $\gamma(A) \leq c(B) \leq c(A)+\Delta$, where the last inequality follows from Proposition \ref{prop:cutproperties}(i). 
\end{proof}

The next step is to characterize the 
{initial resilience} 
$\gamma(I_0)$. The following, {highly nontrivial} theorem from   \cite{BiSe91} {provides an answer when} $I_0=V$. 

\begin{theorem}\label{thrm:monotonepath}
For any graph,  $\gamma(V) = W$
\end{theorem}

Theorem \ref{thrm:monotonepath} {implies that finding a} minimum-width $(V-\emptyset)$-crusade is equivalent to finding {a} minimum-width $(V \downarrow \emptyset)$-crusade.  {This property is not true in general for a general initial set $I_0$; however, a related poperty will be established in Lemma 
\ref{lem:monotonepaths}.}

Next, we explore the connection between the size of a bag and its resilience. {We first}  obtain a  bound on $\gamma(A)$ by considering a crusade which removes all nodes of $A$, {one at a time,}  in an arbitrary order. {We then} obtain a related  bound on $W$ by  constructing a crusade in $\Omega(V-\emptyset)$ {that removes} the nodes of the complement of $A$, {one at a time}, and then {uses} an $A$-optimal crusade $\omega^A$.  These two observations imply {certain constraints (an ``admissible region'')} for the pair $(\gamma(A), |A|)$ on the two dimensional  plane, {which are illustrated} in Figure \ref{fig:region1}.   {Finally,} using the properties of the function $\gamma(\cdot)$ that  have been established so far, we   obtain a refinement of the admissible region, {which is again illustrated}  in Figure \ref{fig:region1}.  

\begin{figure}
\centering
\begin{center}
\includegraphics[width=0.4\textwidth]{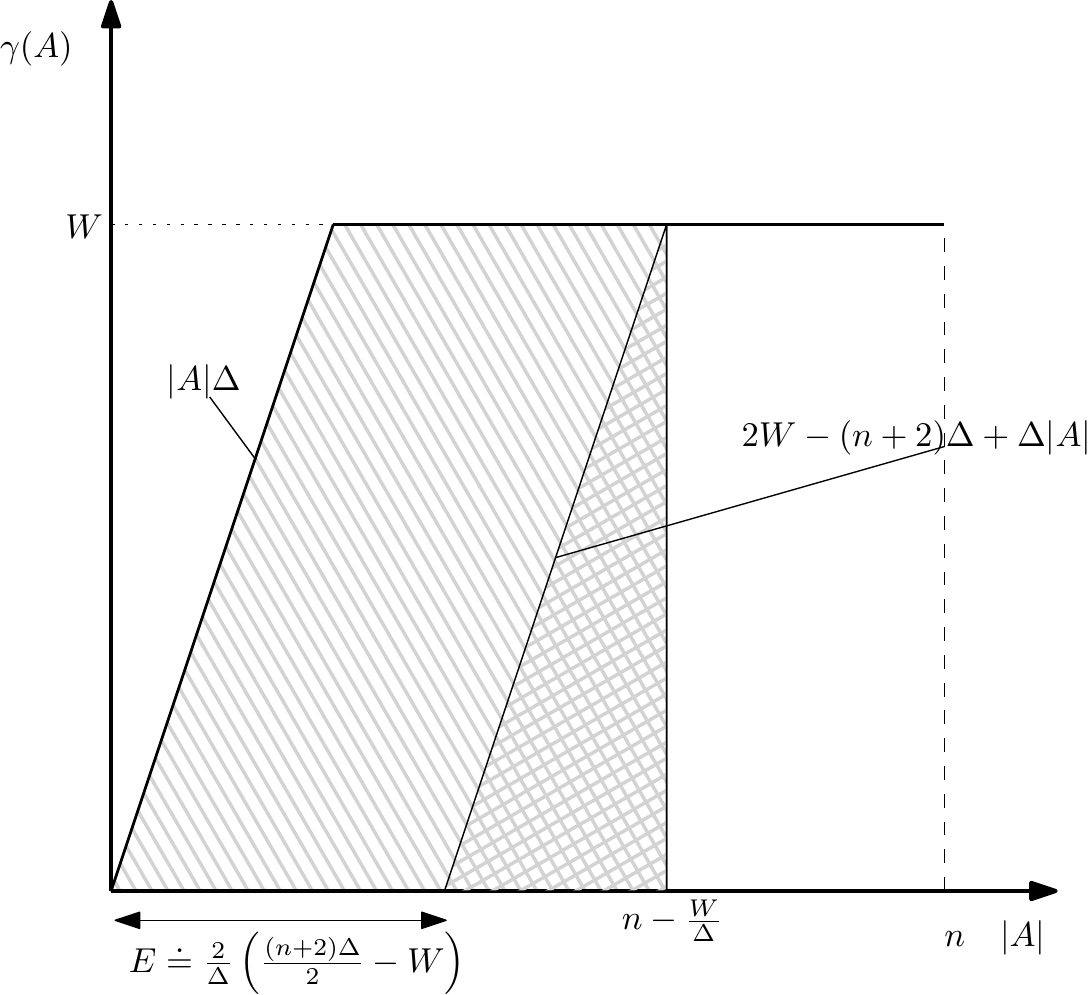}
\caption{{Admissible region for the pair $(\gamma(A),|A|)$. 
{If $\gamma(A)<W$, Lemma \ref{lem:deltasize} implies that $(\gamma(A),|A|)$ belongs to the parallelogram shown in the figure. 
On the other hand, there is no restriction on the size $|A|$ of bags
with $\gamma(A)=W$, and so the admissible region also includes the horizontal} line segment at the top of the figure.  }}\label{fig:region1}
\end{center}
\end{figure}

\begin{lemma}\label{lem:deltasize}
Consider a graph with $W \geq \Delta $ and a bag $A$.   {Let $E$ be as defined in Eq.~\eqref{eq:e}.}
\begin{enumerate}[(i)]
\item $\gamma(A) \leq |A| \Delta$;
\item {If $\gamma(A)<W$, then} $W\leq (n-|A|) \Delta;$
\item {If $\gamma(A)<W$, then}
$\gamma(A) \geq  {\Delta}(|A|-E)$.
\end{enumerate}
\end{lemma} 

\proof
\vspace{0.5cm}

\begin{enumerate}[(i)]

\item Consider some enumeration $(a_1,a_2,\ldots, a_{|A|})$ of the nodes of $A$. We construct {a crusade $\homega\in\Omega(A-\emptyset)$ by letting}  $\homega_0=A$, and $\omega_i=\omega_{i-1} \setminus \{a_i\}$ for $i=1, \ldots, |A|$.   By Proposition \ref{prop:cutproperties} (i),  the maximum cut encountered by $\homega$ is bounded by $|A| \Delta$. Therefore, 
{
$$\gamma(A)\leq z(\homega)\leq |A|\Delta.$$}

\item Consider some  enumeration $(a^c_1, a^c_2, \ldots, a^c_{n-|A|})$ of {the} nodes {of} $A^c$,  the complement of $A$. Let $\omega^A=(\omega^A_0, \ldots, \omega^A_k)\in \Omega^A$.
We construct {a crusade $\omega\in\Omega(V-\emptyset)$ by letting}  $\omega_0=V$, $\omega_i=\omega_{i-1}\setminus\{a^c_i\}$ for $i=1, \ldots, n-|A|$, and $\omega_i=\omega^A_{i-n+|A|}$ for  $i=n-|A|+1, \ldots, k+n-|A|$. Then,
\begin{align}
W{=\gamma(V)}\leq z(\omega) &=\max\{\gamma(A), \max_{{i\in\{1,\ldots,  n-|A|\}}} {\cut(\omega_i)} \} \label{eq:cas} \\ 
&\leq \max\{\gamma(A), (n-|A|)\Delta\}. \nonumber
\end{align}
{The first equality above follows from  Theorem~\ref{thrm:monotonepath}; }
{the second equality follows from the construction of $\omega$;  the last inequality follows} from Proposition \ref{prop:cutproperties}(i).  {Using the assumption} $\gamma(A)<W$, {Eq.~(\ref{eq:cas}) implies that}
\[
W \leq (n-|A|) \Delta.
\]

\item Consider some bag $A$ for which $\gamma(A)<W$. From (ii), 
\begin{equation}
|A|\leq n-W/\Delta.
\label{eq:Bsize}
\end{equation}
Let $C \subseteq V\setminus A$ be some {nonempty} bag with $|C|= n  - \lfloor W/\Delta \rfloor  -  |A| +1$. 
Note that {Eq.~\eqref{eq:Bsize} implies that}
\[
|C| = n- \lfloor W/ \Delta \rfloor  -|A|+1\geq n-  W/ \Delta   -|A|+1 \geq 1,
\]
and {that the assumption $W \geq \Delta$ implies that}
\[
|C| = n- \lfloor W/ \Delta \rfloor  -|A|+1 \leq n-|A|.
\]
{This shows the existence of  a bag with the desired properties exists.}

We define $F=A\cup C$. Note that  
\begin{align*}
|F|&=|A|+|C|\geq n- \lfloor W/\Delta\rfloor +1 \\
&>  n-\lfloor W/\Delta\rfloor \geq n-W/\Delta,
\end{align*}
 and thus
\[
W>(n-|F|)\Delta.
\]
{Then, part (ii) of the Lemma implies} that $\gamma(F)=W$. 
The resilience of $F$ {satisfies}
\[
W=\gamma(F)=\gamma(A \cup C) \leq \gamma(A) + |C| \Delta,
\]
where the inequality follows from applying Lemma \ref{lem:infectiondelta}(ii) $|C|$ times. {Therefore,} 
\begin{align*}
\gamma(A)&\geq W-|C| \Delta \\
&= W-(n- \lfloor W/\Delta \rfloor  -|A|+1) \Delta\\
		& \geq W-(n- W/\Delta +1  -|A|+1)\Delta \\
		& \geq 2W -(n+2)\Delta +\Delta|A| \\
		& {=\Delta(|A|-E),}
\end{align*}
which concludes the proof. 
\end{enumerate} \endproof

\subsection{Characterization of optimal crusades and {some} implications}
\label{s:properties}

In this section we prove that {when} $E$ is small, i.e., {when the}  CutWidth is close to the largest possible value, then bags with large resilience also have large cuts.

\begin{lemma}\label{lem:expcut}
Suppose that 
$W\geq\Delta$ and that the bag $A$ satisfies $0<\gamma(A)<W$.
Then,
\[
\cut(A) \geq \gamma(A) - 2(E+2) \Delta.
\]
\end{lemma}

The rest of the  section is devoted to proving this property.  We start with a  characterization of optimal crusades for a given bag $A$. Specifically, {Lemma \ref{lem:monotonepaths} states} that for any bag $A$, there exists an optimal crusade which: {(i)  can add nodes, and potentially remove one node at the first step; (ii) cannot add nodes (i.e., is monotone) after the first step}
  (parts (i)-(ii)). } Moreover, we argue that {except for} trivial cases, an improvement bag must be encountered before the end of the crusade (part {(vi)}). {These properties allow us to make a connection between resilience and cuts.}

\begin{lemma}\label{lem:monotonepaths}
For any  {nonempty} bag $A$ with $\gamma(A)>0$, there exists a crusade $\hat{\omega}
{=(\homega_0,\homega_1,\ldots,\homega_{k})} \in \Omega^A$ with the following properties:
\begin{enumerate}[(i)]
\item For  $i \in \{1, \ldots, k\} $, ${\homega}_i \neq \homega_{i-1}$
\item For  $i \in \{2, \ldots, k\} $, $\hat{\omega}_{i} \subset \hat{\omega}_{i-1}$.
\item For  $i \in \{0, \ldots, k\} $, $\gamma(\homega_i) \leq \gamma(A)$. 
\item $ \gamma(\homega_1) \geq \gamma(A)-\Delta$. 
\item $ c(A) \geq c(\homega_1)- \Delta(E+2)$.
\item
{Let} $l= {\min} \{ i \geq 0 : \homega_i \in \mathcal{C}\}$. 
Then,  $l<\infty$. 
\end{enumerate}
\end{lemma}

\proof{} 
We assign to  every $(A-\emptyset)$-crusade  $\omega \in \Omega^A$ a value $P(\omega)= \left ( \sum_{i=0}^{|\omega|-1} (\cut(\omega_i)+1), \sum_{i=0}^{|\omega|-1} |\omega_i| \right)$.  Let $\homega \in \argmin_{\omega \in \Omega^A} P(\omega)$, where {the} minimum is taken with respect to the lexicographic ordering.  

\begin{enumerate}[(i)]
\item  \noindent We first prove that  for all $i \in  \{0, \ldots, k-1\}$,
 \begin{equation}\label{eq:notequal}
 \homega_i \neq \homega_{i+1}.
 \end{equation}
For the purposes of contradiction, assume that for some $q \in \{1, \ldots, k-1\}$, $\homega_q = \homega_{q+1}$,  and construct a crusade  $\tilde{\omega}=(\tilde{\omega}_0, \ldots, \tilde{\omega}_{k-1})$ by setting $\tilde{\omega}_i=\homega_i$ for all $i \leq q$, and $\tilde{\omega}_{{i}}=\homega_{{i+1}}$ for ${i=q+1,\ldots,k-1}$.

Clearly, $\tilde{\omega}{=(\tilde{\omega}_0,\ldots,\tilde{\omega}_{{k-1}})}$   is a crusade, i.e., $\tilde{\omega} \in \Omega (A -\emptyset)$. Moreover, $\tilde{\omega} \in \Omega^A$, {because} $\max_{1 \leq i \leq k-1} \cut(\tilde{\omega}_i) = z(\homega)=\gamma(A)$. But  $\sum_{i=0}^{{k-1}} (\cut(\tilde{\omega}_i)+1)<\sum_{i=0}^{{k}} (\cut(\homega_i)+1)$,
{which implies that} 
$P(\tilde{\omega})<P(\homega)$, {and} contradicts the minimality of $\homega$.

 \item The  idea of the proof of this property  is borrowed from \cite{BiSe91}, and is based on the submodularity of $\cut(\cdot)$. 
\noindent We first  argue that for all $i \in \{1, \ldots, {k-1}\}$, 

 \begin{equation} 
 \cut(\homega_{i+1}\cup \homega_i) \geq \cut(\homega_i).
 \label{eq:firstep}
 \end{equation} 

For the purposes of contradiction, assume that there exists some $q \in  \{1, \ldots, k-1\}$ such that  
\begin{equation}\label{eq:cutcont}
 \cut(\homega_{q+1}\cup \homega_q) < \cut(\homega_q),
\end{equation}
and construct the sequence of bags  $\tilde{\omega}=(\tilde{\omega}_0, \ldots, \tilde{\omega}_{k})$, by setting $\tilde{\omega}_i=\homega_i$ for all $i \neq q$ and $\tilde{\omega}_q=\homega_{q+1}\cup \homega_q$. 

We first claim that $\tilde{\omega}$ is a crusade, i.e., $\tilde{\omega} \in \Omega(A-\emptyset)$.  Indeed, since $\homega$ is a crusade, we get $|\homega_q \setminus \homega_{q+1}| \leq 1$ and $|\homega_{q-1} \setminus \homega_{q}| \leq 1$. Therefore, 
\[
|\tilde{\omega}_{q-1} \setminus \tilde{\omega}_q|=|\homega_{q-1} \setminus (\homega_{q+1}\cup \homega_q)| \leq |\homega_{q-1}- \homega_q| \leq 1,
\]
where the first equality follows from the construction of $\tilde{\omega}$ and the second inequality from  $\homega_{q+1}\cup \homega_q$. Furthermore,
\begin{align*}
|\tilde{\omega}_{q} &\setminus \tilde{\omega}_{q+1}| = |(\homega_{q+1}\cup \homega_q) \setminus  \homega_{q+1}|\leq|\homega_{q}- \homega_{q+1}|\leq 1,
\end{align*}
where the the first equality follows from the construction of $\tilde{\omega}$ and the second inequality from  $\homega_{q+1}\cup \homega_q \supset \homega_{q+1}$.

Moreover, we claim that $\tilde{\omega} \in \Omega^A$. 
Indeed 
\begin{align*}
\max_{1 \leq i \leq k} & \cut(\tilde{\omega}_i) = \max \{ c( \tilde{\omega}_q), \max_{1 \leq i \leq k, i \neq q} \cut(\homega_i)\}\\
& \leq \max_{1 \leq i\leq k} \cut(\homega_i) =\gamma(A),
\end{align*}
where the  inequality follows from (\ref{eq:cutcont}). 

On the other hand, it follows from (\ref{eq:cutcont}) that $\sum_{i=0}^{{k}} (\cut(\tilde{\omega}_i)+1)<\sum_{i=0}^{{k}} (\cut(\homega_i)+1)$ and thus $P(\tilde{\omega})<P(\homega)$,
which contradicts  the minimality of $\homega$. We {have} therefore established (\ref{eq:firstep}). 

Using the submodularity of the cut as well as Eq.~(\ref{eq:firstep}), we have that  for all $i \in \{1, \ldots, k-1\}$, 
\begin{equation}\label{eq:hq}
\cut(\homega_{i+1} \cap \homega_i) \leq \cut(\homega_{i+1}).
\end{equation}

We now  prove that $|\homega_{i+1} \cap \homega_i|\geq |\homega_{i+1}|$ for all $i\in \{1, \ldots, k-1\}$.  

For the purposes of contradiction, assume that there exists some $q \in \{1, \ldots, k-1\}$ such that 
\begin{equation}
|\homega_{q+1} \cap \homega_q|< |\homega_{q+1}|.
\label{eq:contr2}
\end{equation}

Construct the sequence $\tilde{\omega}_i=\homega_i$ for all $i \neq q+1$ and $\tilde{\omega}_{q+1}=\homega_{q+1}\cap \homega_q$. 

We first claim that $\tilde{\omega}$ is a crusade, i.e. $\tilde{\omega} \in \Omega(A-\emptyset)$.  Indeed, since $\homega$ is a crusade we get $|\homega_q \setminus \homega_{q+1}| \leq 1$ and $|\homega_{q+1} \setminus \homega_{q+2}| \leq 1$. Therefore, 
\[
|\tilde{\omega}_q \setminus \tilde{\omega}_{q+1}|=|\homega_{q} \setminus (\homega_{q+1}\cap \homega_q)| = |\homega_{q}- \homega_{q+1}|\leq 1,
\]
where the the first equality follows from the construction of $\tilde{\omega}$ and the second inequality from  $\homega_{q+1}\cap \homega_q \subset \homega_{q+1}$. Furthermore,
\begin{align*}
|\tilde{\omega}_{q+1} \setminus \tilde{\omega}_{q+2}|&=|(\homega_{q+1}\cap \homega_q) \setminus  \homega_{q+2}| \\
&\leq  |\homega_{q+1}- \homega_{q+2}|\leq 1,
\end{align*}
where the the first equality follows from the construction of $\tilde{\omega}$ and the second inequality from  $\homega_{q+1}\cap \homega_q \subset \homega_{q+1}$. 
Moreover, we claim that $\tilde{\omega} \in \Omega^A$. 
Indeed 
\begin{align*}
\max_{1 \leq i \leq k} & \cut(\tilde{\omega}_i) = \max \{ c( \tilde{\omega}_q), \max_{1 \leq i \leq k, i \neq q+1} \cut(\homega_{q+1})\} \\
&\leq \max_{1 \leq i\leq k} \cut(\homega_i) =\gamma(A),
\end{align*}

where the  inequality follows from (\ref{eq:hq}). 

On the other hand, it follows from (\ref{eq:hq}) that $\sum_{i=0}^{k} (\cut(\tilde{\omega}_i)+1)\leq\sum_{i=0}^{k} (\cut(\homega_i)+1)$ and from (\ref{eq:contr2}) that  $\sum_{i=0}^{k} |\omega_i| < \sum_{i=0}^{k} |\tilde{\omega}_i|$. Therefore, $P(\tilde{\omega})<P(\homega)$,
which contradicts  the minimality of $\homega$.

Therefore we established that   $|\homega_{i+1} \cap \homega_i|\geq |\homega_{i+1}|$ for all $i\in \{i, \ldots, k-1\}$.  The latter implies that {for} all $i\in \{1, \ldots, k-1\}$, $\homega_{i+1} \subseteq \homega_i$.  Using {part} (i) of the lemma, {it follows that} that for $i \in \{1, \ldots, k-1\}$, $\homega_{i+1} \subset \homega_i$.

\item We prove the result by induction. First, observe that $\gamma(\homega_0)=\gamma(A)$. Assume that $\gamma(A) \geq \gamma(\homega_i)$. Moreover, by (\ref{eq:bellman}), for all $i \in \{1, \ldots, k-1\}$,  $\gamma(\homega_i)=\max\{\gamma(\homega_{i+1}), \cut(\homega_{i+1})\} \geq \gamma(\homega_{i+1})$. Therefore, $\gamma(\homega_{i+1}) \leq \gamma(A)$.

\item We consider two cases. Assume that ${\hat\omega}_1 \supset A$. Then $\gamma(\homega_1) \geq \gamma(A){\geq \gamma(A) - \Delta}$. Otherwise, by the definition of a crusade we get $|A\setminus \homega_1| \leq 1$. Therefore, we can write $\homega_1=A \cup D -v$, {for some set $D$ (disjoint from $A$) and some $v\in A$.} {Using} Lemma \ref{lem:infectiondelta}(ii), {and then Lemma \ref{lem:infectiondelta}(i), we obtain}
$$\gamma(\homega_1) \geq \gamma(A \cup D) -\Delta
{\geq \gamma(A)-\Delta}.$$
 
\item From (iii) we obtain $\gamma(\homega_1) \leq \gamma(A)$. Therefore, using Lemma \ref{lem:deltasize}(iii), we conclude that 
\begin{equation}
|\homega_1| \leq \frac{\gamma(A)}{\Delta} + E. \label{eq:h1}
\end{equation}
Moreover, by Lemma \ref{lem:deltasize}(i), we get 
\begin{equation}
|A| \geq \frac{\gamma(A)}{\Delta}.\label{eq:h2}
\end{equation}
We consider two cases. Assume that $\omega_1 \supset A$. Then,
\[
|\homega_1  {\blacktriangle} A| = |\homega_1|-|A|.
\] 
Otherwise, by the definition of a crusade we get $|A\setminus {\homega}_1| \leq 1$. Therefore, we can write   ${\hat\omega}_1=A \cup D -v$, 
{where $D$ is disjoint from $A$ and $v\in A$. Thus,}
\[
|\homega_1  {\blacktriangle} A| = |D|+1 = |\homega_1| - |A|+2. 
\]

Therefore, in both cases,
\[
|\homega_1  {\blacktriangle} A|  \leq |\homega_1| - |A|+2.
\]
We {then} use Eqs.~(\ref{eq:h1}) and (\ref{eq:h2}) to obtain
\[
|\homega_1  {\blacktriangle} A|\leq E+2.
\]
The result follows by applying Proposition \ref{prop:cutproperties}(iv). 

\item Note that $\gamma(\homega_k)=\gamma(\emptyset)=0$.
Note also that any single-element set $B$ satisfies $\gamma(B)=0$.
Suppose that $\gamma(\homega_1)>0$. Then, there exists some $i\in\{1,\ldots,k-1\}$ such that $\gamma(\homega_{i+1})<\gamma(\homega_i)$,
and $\homega_i$ is an improvement bag.

{Suppose now that $\gamma(\homega_1)=0$. If $\homega_1=\emptyset$, then $\gamma(A)=0$, which contradicts the assumption $\gamma(A)>0$. If $\homega_1$ is nonempty, then we must have $\cut(\homega_2)=0$, so that $\homega_2$ is empty and $\homega_1$ is a singleton. 
Since $\gamma(A)>0$, the set $A$ is not a singleton. Since at most one element can be removed in going from $A$ to $\homega_1$, it follows that $A$ consists of two elements and that a single element was removed from $A$. In that case, $A=\homega_0$ is an improvement bag.}

{In both cases, we see that there exists some $i$ for which $\homega_i$ is an improvement bag and therefore $l$ is well-defined and finite.}
\endproof
\end{enumerate}

\proof{(of Lemma \ref{lem:expcut})}
Consider {a} crusade $\homega_A$ with the properties {in} Lemma \ref{lem:monotonepaths},  and {let} {$l\geq 0$ be such that} $B=\homega^A_l$ is the fist improvement bag encountered.


From Lemma \ref{lem:monotonepaths}(vi), {$l$ is well-defined and finite. We} consider three cases:
\begin{enumerate}[(i)]
\item $l=0$: If $A$ is itself an improvement bag, then from Lemma \ref{lem:cutwhendrops}, $\cut(A) \geq \gamma(A)-\Delta$.
\item $l=1$:   In this case, {$\homega_1$ is an improvement bag.} From Lemma \ref{lem:cutwhendrops}, $\cut(\homega_1) \geq \gamma(\homega_1)-\Delta$. Then, from Lemma \ref{lem:monotonepaths}(iv), we obtain 
\[
\cut(\homega_1) \geq \gamma(A)-2\Delta. 
\]

Moreover, from Lemma \ref{lem:monotonepaths}(v), we get
\[
\cut(A)\geq \cut(\homega_1) - (E+2) \Delta \geq \gamma(A) - (E+4) \Delta.
\]

\item $l>2$: In this case, {by property (ii) in} Lemma \ref{lem:monotonepaths}(ii), 
{it folows that}
$B \subset \homega_1$ {and}
 \[
 |B  {\blacktriangle} \homega_1|= |\homega_1|-|B|.
 \]
 Moreover, since $B$ is the first improvement bag that is encountered, $\gamma(B)=\gamma(\homega_1)$.  We use Lemma \ref{lem:deltasize}(i) to obtain 
 \[
 |B| \geq \gamma(B)/\Delta=\gamma(\homega_1)/\Delta,
 \] 
 and Lemma \ref{lem:deltasize}{(iii) } to obtain
\[
 |\homega_1| \leq \gamma(\homega_1)/\Delta + E.
 \] 
 Combining the above, 
 \[
 |B  {\blacktriangle} \homega_1|= |\homega_1|-|B| \leq E,
 \]
 from which we conclude that
 \[
 \cut(\homega_1) \geq \cut(B)- E \Delta \geq \gamma(\homega_1)-(E+1) \Delta.
 \]
 where the first inequality follows from Proposition \ref{prop:cutproperties}(iv) and the second from the fact that $B$ is an improvement bag and $\gamma(B)=\gamma(\omega_1)$. Therefore, from Lemma \ref{lem:monotonepaths}(v), we obtain
 \[
 \cut(A) \geq \gamma(\omega_1) - (2E+3) \Delta.
 \]
 Finally, using Lemma \ref{lem:monotonepaths}(iv), we {conclude that}
 \[
 \cut(A) \geq \gamma(A) - 2(E+2) \Delta,
 \] 
 \end{enumerate}
 {which completes the proof of Lemma \ref{lem:expcut}.
}

{The combinatorial properties of the resilience derived in this section will be used next to obtain
a lower bound on the expected extinction time,}
in the regime where $\gamma(I_0) \gg r$.

\section{Exponential Lower Bound} \label{sec:exponential}
In this section we {state and} prove our main result. 
{Specifically,}
we use Lemma \ref{lem:expcut} to argue that the 
{process must traverse a region in which the}
number {$|I_t|$} of infected nodes has an upward drift, which in turn {leads to} the desired lower bound.  

\begin{theorem}\label{thrm:exponential}
{Suppose that $\gamma(I_0)\geq \Delta (9E+{12}) + 3r$.} 
Then, 
\[
\mathbb{E}_{I_0}[\tau] \geq \frac{1}{2r}\left( \left( \frac{\gamma(I_0)-(9 E+12) \Delta}{3r} \right)^{\frac{\gamma(I_0)}{3\Delta}-1}-1 \right).
\]
\end{theorem}
\begin{proof}
We define a process $V_t$ which is coupled with the process $I_t$ as follows. 
\begin{equation}
V_t=
\begin{cases}
|I_t|, & \text{ if }  |I_t| \leq \left \lfloor \frac{2\gamma(I_0)}{3\Delta} \right \rfloor,\\
 \left \lfloor \frac{2\gamma(I_0)}{3\Delta}\right \rfloor, & \text{ if }  |I_t| > \left \lfloor\frac{2\gamma(I_0)}{3\Delta}\right \rfloor.
\end{cases}
\end{equation}
The dynamics of $V_t$ {are as follows.} If $ i < \lfloor 2\gamma(I_0)/3\Delta)\rfloor$, then
\begin{align}
V_t:&\ i \rightarrow i+1, \text{  with rate  }  c(I_t), \nonumber\\ 
V_t:&\  i \rightarrow i-1, \text{ with rate  } r.\nonumber
\end{align}
{Furthermore,} if $i = \lfloor 2{\gamma(I_0) }/(3\Delta))\rfloor$, then
\[
V_t: \ i\rightarrow i-1,  \text{  with rate  } {r(I_t)},
\]
where ${r(I_t)} \leq r$. 

{Consider the stopping time}
\begin{align*}
{\tau^*}&= \inf\left\{t\geq 0 : |I_t|\leq \left \lfloor \frac{\gamma(I_0)}{3\Delta} \right \rfloor\right\}
\\
&=\inf\left\{t\geq 0 : V_t\leq \left \lfloor \frac{\gamma(I_0)}{3\Delta} \right \rfloor\right\}.
\end{align*}
For every sample path, ${\tau^*}\leq \tau$.  Therefore, 
\[
\mathbb{E}_{I_0}[\tau] \geq \mathbb{E}_{I_0}[{\tau^*}].
\]

{Suppose now that $|I_t|$ satisfies}
$$ \frac{\gamma(I_0)}{3\Delta} {\leq |I_t| \leq} \frac{2\gamma(I_0)}{3\Delta}.$$
Then, by parts {(i) and (iii) of} Lemma \ref{lem:deltasize}, we obtain 
that  
$$ \frac{\gamma(I_0)}{3}-E \Delta {\leq \gamma(I_t) \leq} \frac{2\gamma(I_0)}{3}.$$
Furthermore,  
{Lemma \ref{lem:expcut} implies that}
\[
c(I_t) \geq \gamma(I_t)-2(E+2)\Delta \geq \frac{\gamma(I_0)}{3}-(3E+4) \Delta.
\]
{It follows that} the process $V_t$ stochastically dominates a process $Y_t$, {described in the Appendix}, with parameters $\lambda=r$, $\mu=\gamma(I_0)/3-(3E+4)\Delta$, and $L=\gamma(I_0)/3\Delta$.
{Therefore, using {Eq.~\eqref{eq:lowerrandom}} in the Appendix,}
\[
\mathbb{E}_{I_0}[\tau] \geq \frac{1}{2r}\left( \left( \frac{\gamma(I_0)-(9 E+12) \Delta}{3r} \right)^{\frac{\gamma(I_0)}{3\Delta}-1}-1 \right).
\]
\end{proof}

Note that when $3r<\gamma(I_0)-9E\Delta {-12\Delta}$, the optimal expected extinction time scales exponentially in the resilience of the set of the  initially infected nodes.
 In \cite{DOTcdc14} and \cite{DOTtnse14}, we {focused on} the case where $I_0=V$ ({the} worst case) and proved that {if} the CutWidth {of the graph} is a sublinear function of the number of nodes, {and} if $r=o(n)$, {then, the} expected extinction time  is $o(n)$. In contrast, the following result considers the case where $W$ scales linearly  in the number of nodes and provides an exponential lower bound on the expected extinction time. 
Specifically,  using Theorem \ref{thrm:monotonepath} {to replace $\gamma(V)$ by $W$, and using also the definition of $E$,} we can write our lower bound as
\[
\mathbb{E}_{V}[\tau] \geq \frac{1}{{2}r}\left( \left( \frac{19 W-9 n \Delta {30}\Delta }{3r} \right)^{W/3\Delta-1}-1 \right),
\]
{and we obtain the  corollary that follows.}

\,
\,

\begin{corollary}
	{Fix 	{a} constant $C>1$ and consider those graphs for which}
	{
	 $$W\geq\frac{9C}{19}  n\Delta.$$} Moreover, assume that $r=o(n)$. Then,
	\[
\mathbb{E}_{V}[\tau] =\Omega (2 ^ n).
\]
\end{corollary}

 \bibliography{epidemics}
 \bibliographystyle{abbrv}

\appendix
{We first recall a basic result on} the standard continuous random walk on the integers. {Specifically, let $Z_t$ denote the state of a Markov process with the following dynamics:
\begin{align}
&Z_t: i \rightarrow i+1, \,\, \text{with rate} \,\,\mu,  \nonumber \\ 
&Z_t: i \rightarrow i-1, \,\,  \text{with rate}\,\, \lambda.
\label{eq:ytprocess}
\end{align}}

{Fix some integers $M$ and $L$ with}
$0<M<L$,  
{and let $\mathbb{P}_{M}$ be the probabilty measure that describes the process when initialized at $Z_0=M$.}
Let $$\tau_{L}=\inf\{t: Z_t=0 \,\, \text{or}\,\, Z_t=L\}$$ denote the first time that the process $Z_t$ visits state $0$ or $L$, {which is}  a \emph{stopping time}. Moreover, let $$\hat{\tau}=\inf\{t: Z_t=0 \}$$ denote the {first} time that $Z_t$ hits $0$. 
 
{The following result is standard; see, e.g.,  Section 2.1 of \cite{Peres} 
or Section 2.3 of \cite{Stee00}.}

\begin{lemma}\label{lem:randwalk}
Consider the process $Z_t$ and the stopping times $\tau_L$ and $\hat{\tau}$. Then,
\begin{equation}\label{eq:upprob}
\mathbb{P}_{M}({Z}_{\tau_{L}}=L)=\frac{1-(\lambda/\mu)^M}{1-(\lambda/\mu)^L},
\end{equation}
\end{lemma}

{Consider now a related} Markov process $Y_t$, whose {transition rates are} as follows:
\begin{align}
&Y_t: i \rightarrow i+1, \,\, \text{with rate} \,\,\mu,  \nonumber \\ 
&Y_t: i \rightarrow i-1, \,\,  \text{with rate}\,\, \lambda,
\label{eq:ytprocess}
\end{align}
for $i \in \{1,\ldots,L-1\}$ while
\[
Y_t: i \rightarrow i-1, \,\,  \text{with rate}\,\, \lambda,
\]
for $i=L$ and
\[
Y_t: i \rightarrow i+1, \,\,  \text{with rate}\,\, \mu,
\]
for $i=0$.

{We are looking for a lower bound on the expected time that it takes for the process $Y_t$ to hit $0$ for the first time, assuming that it starts at $L-1$.
Let ${p}$ be the  probability that $Y_t$ hits level $L$ before hitting $0$ starting from state $L-1$, which is given by Lemma \ref{lem:randwalk}, with $M=L-1$. We consider the case where $\lambda<\mu$, so that $p>1/2$. 
Each time that the process is at state $L-1$, the process regenerates, and we have a new trial, which succeeds in hitting state 0 before state $L$, with the same probability $1-p$. Let $N$ be the number of trials and note that its expected value is 
$1/(1-p)$. In between trials, there needs to be a transition from state $L$ to state $L-1$, whose expected time is $1/\lambda$. Thus, the total expected time elapsed until state 0 is hit for the first time is $\mathbb{E}[N-1]/\lambda=p/(1-p)\lambda$. Using Lemma \ref{lem:randwalk} and some straightforward algebra, we obtain that this expected time is at least as large as
\begin{equation}\label{eq:lowerrandom}
\frac{1}{2}\left( \left( \frac{\mu}{\lambda} \right)^{L-1}-1 \right)\frac{1}{\lambda}.
 \end{equation}
}

\end{document}